\documentclass[12pt]{article}

\usepackage{amsthm,amssymb,amsmath,graphicx}
\usepackage{color}
\usepackage[top=2cm, bottom=2cm, left=2.5cm, right=3cm]{geometry}
\usepackage[pagebackref=false,colorlinks,linkcolor=blue,citecolor=magenta]{hyperref}
\usepackage{multicol}
\usepackage{tikz}
\usepackage{algcompatible}
\usepackage{subcaption}
\usepackage{algorithm}
\usepackage{algpseudocode}
\usepackage{tcolorbox}
\usepackage{pgfplots}
\usepackage{cite}
\usepackage{amsmath,amssymb,amsfonts,amsthm}
\usepackage{graphicx}
\usepackage{textcomp}
\usepackage{xcolor}
\def\BibTeX{{\rm B\kern-.05em{\sc i\kern-.025em b}\kern-.08em
    T\kern-.1667em\lower.7ex\hbox{E}\kern-.125emX}}

\newcommand{\ket}[1]{|#1\rangle}
\newcommand{\bra}[1]{\langle #1|}

\newcommand{\ketbra}[2]{\ket{#1}\!\bra{#2}}

\usepackage[english]{babel}

\newtheorem{lemma}{Lemma}[section]

\newtheorem{theorem}{Theorem}[section]
\newtheorem{definition}{Definition}[section]

\newtheorem{remark}{Remark}[section]

\newcommand{\mc}[1]{\mathcal{#1}}
\newcommand{\tr}{\mathrm{tr}}

\newcommand{\cT}{\mc T}
\newcommand{\cE}{\mc E}

\newcommand{\cM}{\mc M}

\newcommand{\cD}{\mc{D}}

\newcommand{\cP}{\mc{P}}

\newcommand{\idd}{\operatorname{id}}

\usepackage[affil-it]{authblk}

\author[1]{Omar Fawzi}
\author[2,3]{Li Gao}
\author[4,5]{Mizanur Rahaman}
\author[1]{Mostafa Taheri}

\affil[1]{Univ Lyon, Inria, ENS Lyon, UCBL, LIP, Lyon, France}

\affil[2]{School of Mathematics and Statistics, Wuhan University, Wuhan, China}

\affil[3]{Wuhan Institute of Quantum Technology, Wuhan, China}

\affil[4]{Department of Mathematical Sciences, Chalmers University of Technology and University of Gothenburg, Gothenburg, Sweden}

\affil[5]{Wallenberg Centre for Quantum Technology, Chalmers University of Technology, Gothenburg, Sweden}
\begin{document}

\title{Fast convergence of Dynamic Capacities of GNS-Symmetric Quantum Channels
\thanks{MT and OF acknowledge funding by the European Research Council (ERC Grant AlgoQIP, Agreement No. 851716), LG by the National Natural Science Foundation of China (NNSFC grant No. 12401163) and by the Department of Science and Technology of Hubei Province (Project No. 2025EHA041, No. 2025AFA044). This work started when MR was a Marie-Sk\l odowska-Curie Fellow at ENS de Lyon, and he acknowledges funding from the European Union’s Horizon Research and Innovation programme (Agreement No. HORIZON-MSCA-2022-PF-01, project No. 101108117). This work was partially
supported by the Wallenberg Centre for Quantum Technology (WACQT) funded by Knut and
Alice Wallenberg Foundation (KAW). }
}

\maketitle

\begin{abstract}
 We consider a quantum system described by a quantum channel $\Phi$ that is applied at every time step and study the time evolution of its information capacities. When $\Phi$ is a GNS-symmetric channel (this includes Pauli channels, for example), we give explicit exponential convergence bounds for the classical and quantum capacities. These bounds are in terms of entropic properties of $\Phi$. We further illustrate how these results help quantify the performance of active versus passive error-correction setups.
\end{abstract}

\section{Introduction}

Consider a finite-dimensional quantum system that we would like to use as a memory. 
The system is subject to noise, which we model by a channel $\Phi$ that is applied at each discrete time step. We encode $\log D$ bits or qubits of information at time zero, let the system evolve under $\Phi$ for $t$ time steps, and then apply a decoding map in the hope of recovering our message. A basic question is: how long can we retain information, and how does the size of the message depend on the noise? In other words, what is the minimum error achievable when storing $\log D$ (qu)bits for the channel $\Phi^t$ which is the $t$-fold sequential composition of $\Phi$, and how does this behave as $t\to\infty$?

\medskip\noindent\textbf{Prior work.} When $t=1$, this corresponds to the fundamental question in quantum Shannon theory and has been studied in great depth~\cite{wilde2013quantum}. For large $t$, even though the question has been treated in depth from various viewpoints such as self-correcting quantum memories~\cite{brown2016quantum} or quantum memories with active error correction~\cite{terhal2015quantum}, it is only recently that it has been attacked from a Shannon theory perspective with the work of Singh, Rahaman and Datta~\cite{singh2024zero}. They studied the zero-error regime and showed that if $\Phi$ is mixing (respectively asymptotically entanglement-breaking), then no amount of classical (respectively quantum) information can be recovered perfectly at late times. They derived universal quadratic bounds on the scrambling time: in a $d$-dimensional system, perfect information survives for at most $O(d^2)$ steps, and optimal encodings must live inside the peripheral subspace of the channel, a direct sum of matrix blocks corresponding to peripheral eigenvalues. In the non-scrambling case, the same authors and independently Fawzi, Rahaman and Taheri~\cite{fawzi2024capacities} obtained closed-form expressions: the asymptotic zero-error classical capacity equals $\log \sum_k d_k$ and the quantum capacity equals $\log \max_k d_k$, where $d_k$ depend on the algebraic properties of the peripheral subspaces. Moreover,~\cite{fawzi2024capacities} extended these results to the regime where a non-zero error is allowed. They proved that, for any fixed error $\delta\in[0,1)$, the classical capacity of $\Phi$ as $t\to\infty$ is given by
\[
C_{\infty,\delta}(\Phi)
   = \log\left\lfloor\frac{\sum_k d_k}{1-\delta}\right\rfloor,
\]
while the quantum capacity satisfies
\begin{equation}
Q_{\infty,\delta}(\Phi)
\approx
\log \Bigl(\frac{\max_k d_k}{1-\delta}\Bigr).
\end{equation}
These results imply that, in the infinite-time limit, approximate capacities coincide (up to small additive terms) with their zero-error counterparts and depend solely on the peripheral subspace. Moreover, the infinite-time capacities are additive under tensor products and can be computed efficiently by determining the structure of the peripheral subspace.

More recently, Singh and Datta \cite{singh2024information,singh2025informationstoragetransmissionmarkovian} analysed the finite-time $\delta$-error capacities for general Markovian noise. They derived upper and lower bounds on classical, private classical, entanglement-assisted and quantum capacities of $\Phi^t$ that converge exponentially to the infinite-time formulas when $t$ exceeds a threshold that is quadratic in the input dimension of $\Phi$. 

The purpose of this work is to determine simple conditions under which the information capacity converges very quickly (e.g., in time scaling at most logarithmically in the dimension of the system) to the infinite-time capacity.

\medskip\noindent\textbf{Our contributions.} We focus on GNS-symmetric quantum channels and semigroups, a class of noise models that enjoy a reversible structure with respect to a full-rank invariant state. For such channels, powerful functional inequalities are available; see e.g.,~\cite{rouze2024logarithmic}. Our main results show that the classical and quantum capacities of $\Phi^{2t}$ (and of continuous-time semigroups $\cT_t$) converge exponentially fast to their infinite-time limits. The bounds are explicit, hold for any $t \geq 1$ and depend on entropic properties of $\Phi$ that can in turn be bounded by the spectral gap and minimal eigenvalue of a full-rank fixed point of $\Phi$.  More precisely, if the peripheral subspace of $\Phi$ decomposes as 
\[
\mathcal X_\Phi = \bigoplus_k \mathcal{M}_{d_k} \otimes \omega_k,
\]
then we prove the following:

\begin{itemize}

\item\textbf{One-shot capacities decay exponentially.}
The $\delta$-error classical capacity of $\Phi^{2t}$ satisfies
\[
C_\delta(\Phi^{2t}) \le
\frac{\log\bigl(\sum_k d_k\bigr) 
      + e^{-2\alpha_c t}\,\log \Lambda_c(\mathcal{P}) + h(\delta)}{1-\delta},
\]
and, if $\delta<1/4$, the $\delta$-error quantum capacity satisfies 
\[
Q_\delta(\Phi^{2t}) \le
\frac{\log(\max_k d_k)
      + e^{-2\alpha_c t}\,\log \Lambda_c(\mathcal{P})
      + 2 h(\delta)}{1-4\delta}.
\]
Here $\alpha_c$ is the complete entropy contraction coefficient of $\Phi$, $\mathcal{P}$ is the peripheral projection, $\Lambda_c$ is the Pimsner--Popa constant, and $h(\delta)$ denotes the binary entropy. These bounds show that the capacities approach $\log(\sum_k d_k)$ and $\log(\max_k d_k)$ at an exponential rate $e^{-2\alpha_c t}$.

\item\textbf{Asymptotic capacities converge exponentially.}
For the asymptotic classical and quantum capacities we obtain
\[
\log \bigl(\sum_k d_k\bigr)
\!\le\!
C(\Phi^{2t})
\!\le\!
\log \bigl(\sum_k d_k\bigr)
 + e^{-2\alpha_c t}\log \Lambda_c(\mathcal{P}),
\]
and
\begin{align*}
\log(\max_k d_k)
\;& \le\;
Q(\Phi^{2t})
\\ &\le
\log(\max_k d_k)
  + e^{-2\alpha_c t}\,\log \Lambda_c(\mathcal{P}).
\end{align*}
The same bounds hold for continuous-time semigroups with $\Phi^{2t}$ replaced by $\cT_t$. These results provide explicit exponential decay of the gap to the infinite-time capacities for any $t \geq 1$. Compared to the bounds obtained in~\cite{singh2025informationstoragetransmissionmarkovian}, our bounds are explicit and we do not need to take $t$ large enough but our bounds only hold for GNS-symmetric channels.

\item\textbf{Zero-error capacities stabilise exponentially fast.}
We also show that, for GNS-symmetric channels, we obtain an exponential mixing-time threshold for perfect information transmission. 
 In particular, for tensor powers $\Phi^{\otimes n}$ with $\Phi$ GNS-symmetric, this threshold scales linearly in $n$, yielding an exponential improvement over the $2^{O(n)}$ convergence bound of~\cite{singh2024zero}, which holds for all channels.

\item\textbf{Implications for active versus passive error correction.}
Combining our passive upper bounds with lower bounds from concrete coding schemes yields explicit crossover times beyond which active error correction provably outperforms any passive strategy. In Section~\ref{sec: examples} we formalize the $n$-active hierarchy of recovery operations, show that its ceiling coincides with the  asymptotic quantum capacity $Q(\Phi)$, and apply the resulting comparison recipe to Pauli noise using the 5-qubit code and its recursive concatenation into a 25-qubit scheme.

\end{itemize}

\section{Definitions and Preliminaries}
\label{sec: definitions}
\noindent {\bf Notations.} We denote by $B(H)$ the set of bounded operators on a Hilbert space $H$. Hilbert spaces are indexed by capital letters, such as $H_A, H_B, H_E, \ldots$. We denote by $\mathcal{M}_d$ the algebra of $d \times d$ complex matrices. The identity operator is written as ${\bf 1}$, and the identity map as $\mathrm{id}$. The adjoint of an operator $X$ is denoted by $X^\dagger$, and the adjoint of a map $\Phi$ (for the Hilbert-Schmidt inner product) by $\Phi^*$. We use $\Phi^t$ to denote the $t$-fold composition of the channel $\Phi$ with itself, and $\Phi^{\otimes n}$ to denote the $n$-fold tensor product of $\Phi$.
For a state $\rho$, $H(\rho):=-\tr(\rho\log\rho)$ denotes the von Neumann entropy; for a probability vector $p$, $H(p):=-\sum_i p_i\log p_i$ denotes the Shannon entropy. The binary entropy is written as $h(\delta):=-\delta\log\delta-(1-\delta)\log(1-\delta)$.

\subsection{Peripheral subspace and projection}

For a quantum channel $\Phi: B(H) \to B(H)$, the \textbf{peripheral subspace}~\cite{Wolfe2012} is defined as  
\[
span\{X : \exists \lambda \in\mathbb{C},\ \lvert\lambda\rvert=1,\ \Phi(X)=\lambda X\}.
\]

The \textbf{peripheral projection} of the channel $\Phi$, denoted by $\mathcal{P}_\Phi$ (or simply $\mathcal{P}$ when clear from context), is a projection satisfying the following properties~\cite[Theorem 6.16 and Proposition 6.3]{Wolfe2012}:
\begin{enumerate}
    \item $\mathcal{P}_\Phi(B(H)) = \mathcal X_\Phi$,
    \item there exists a sequence of non-decreasing integers $\{n_i\}$ such that 
    \(
    \lim_{i \to \infty} \Phi^{n_i} = \mathcal{P}_\Phi,
    \)
    \item $\mathcal X_\Phi$ is the fixed-point subspace of $\mathcal{P}_\Phi$, i.e.,
    \(
    \mathcal{P}_{\Phi}(X) = X, \quad \forall X \in \mathcal X_\Phi.
    \)
\end{enumerate}

There exists a decomposition of the Hilbert space $H$ of the form  
\(
H = H_0 \oplus \bigoplus_{k=1}^K (H_{k,1} \otimes H_{k,2}),
\)
such that  
\begin{equation}
\label{eq: peripheral structure}
\mathcal X_\Phi = 0 \oplus \bigoplus_{k=1}^K \mathcal{M}_{d_k} \otimes \omega_k,
\end{equation}
where each $\omega_k$ is a fixed state in $B(H_{k,2})$.

As we will see in the following sections, this Hilbert space decomposition and the peripheral projection play an important role in the study of the channel’s dynamical capacity.

\subsection{ Capacity of Quantum Channels}

There are several relevant capacities we consider. In the one-shot setting, we ask how much (classical or quantum) information can be transmitted using a single use of the channel, while allowing for a small error.

Let $\Phi: B(H)\to B(H)$ be a quantum channel. We say that a quantum channel $\cE: \cM_d\to B(H)$ is
\begin{enumerate}
    \item a $(d,\delta)$ classical code if there exists a quantum channel $\cD: B(H)\to \cM_d$ such that 
    \begin{equation}
        F_c( \cD \circ \Phi \circ \cE):=\frac{1}{d}\sum_{i=1}^d \langle i| \cD \circ \Phi \circ \cE(|i\rangle\langle i|)|i\rangle \geq 1 - \delta,
    \end{equation}
    \item a $(d,\delta)$ quantum code if there exists a quantum channel $\cD: B(H)\to \cM_d$ such that 
    \begin{equation}
       F_e(\cD \circ \Phi \circ \cE):= \langle \phi | (\idd \otimes \cD \circ \Phi \circ \cE)(|\phi\rangle\langle \phi|) | \phi \rangle \geq 1 - \delta,
    \end{equation}
    
     where $|\phi\rangle = \frac{1}{\sqrt{d}} \sum_{i=1}^d |i\rangle \otimes |i\rangle$ is the maximally entangled state over $\mathbb{C}^d\otimes \mathbb{C}^d$.
\end{enumerate}
\begin{definition}
    The one-shot $\delta$-error classical and quantum capacities of a channel $\Phi$ are defined as the supremum of $\log d$ over all $(d,\delta)$ codes for $\Phi$, namely
    \begin{align*}
&C_\delta(\Phi) = \sup \{ \log d \ :\ \exists \ (d,\delta) \text{ classical code } \cE \}\\
&Q_\delta(\Phi) = \sup \{ \log d \ :\ \exists \ (d,\delta) \text{ quantum code } \cE \}
\end{align*}
\end{definition}

\begin{definition}[Asymptotic capacities] The asymptotic capacities of a channel $\Phi$ are defined as 
\begin{align*}
    C(\Phi) 
    = \sup \Big\{ & R \geq 0 : \forall \delta>0, \exists n\ge1 ,  R \leq \frac{1}{n}C_{\delta}(\Phi^{\otimes n}) \Big\}
\end{align*}
and similarly for $Q(\Phi)$.
\end{definition}

\begin{remark}
The two capacities defined above admit information-theoretic characterizations in terms of regularized entropic quantities~\cite{devetak2005private,holevo1998capacity,schumacher1997sending,shor2003capacities,lloyd1997capacity}. Specifically,
\begin{align*}
C(\Phi) &= \lim_{n \to \infty} \frac{1}{n} \, \chi\!\left(\Phi^{\otimes n}\right), \\
Q(\Phi) &= \lim_{n \to \infty} \frac{1}{n} \, I_c\!\left(\Phi^{\otimes n}\right),
\end{align*}
where
\begin{align*}
&\chi(\Phi) = \sup_{\rho_{XA}} I(X:B)_{(\mathrm{id}_X \otimes \Phi)(\rho)}, \\
&I_c(\Phi) = \sup_{\rho_{A'A}} I(A' \rangle B)_{(\mathrm{id}_{A'} \otimes \Phi)(\rho)}.
\end{align*}
Here, the first supremum is taken over all classical--quantum states of the form
\[
\rho_{XA} = \sum_x p(x) \, \ketbra{x}{x} \otimes \rho_x,
\]
and the second supremum is over all (pure) bipartite states 
$\rho_{A'A} \in B(H_{A'} \otimes H_A)$ with $H_{A'} \cong H_A$.
Here, for any bipartite state $\omega_{AB}$, the quantum mutual information 
and the coherent information are defined by
\begin{align*}
I(A:B)_\omega &= H(\omega_A) + H(\omega_B) - H(\omega_{AB}), \\
I(A\rangle B)_\omega &= H(\omega_B) - H(\omega_{AB}),
\end{align*}
where $\omega_A = \tr_B \omega_{AB}$ and $\omega_B = \tr_A \omega_{AB}$ are the marginals.  The complementary channel 
$\Phi^c(\cdot) = \operatorname{Tr}_B[V(\cdot)V^\dagger]$ is defined via any Stinespring dilation 
$\Phi(\cdot) = \operatorname{Tr}_E[V(\cdot)V^\dagger]$ of $\Phi$.

\end{remark}

\subsection{GNS-symmetry}
GNS-symmetry captures a natural notion of reversibility for a quantum channel with respect to a full-rank state $\sigma$. It is defined using the $\sigma$-weighted inner product induced by the GNS construction and has proven useful for analyzing the structure and convergence of quantum Markov processes. 
We refer to \cite{GR,GJLL} for further details on GNS-symmetric quantum channels and their associated functional inequalities.
\begin{definition}
Let $\Phi : B(H) \to B(H)$ be a quantum channel and let $\sigma$ be a full-rank state on $H$. 
We say that $\Phi$ is \emph{GNS-symmetric} with respect to $\sigma$ if
\[
\tr\!\left( X\, \Phi^*(Y)\, \sigma \right)
=
\tr\!\left( \Phi^*(X)\, Y\, \sigma \right),
\qquad \forall\, X,Y \in B(H).
\] 
It is known that if $\Phi$ is GNS-symmetric with respect to $\sigma$, then $\sigma$ is an invariant state, i.e., $\Phi(\sigma) = \sigma$. 

We say that a quantum Markov semigroup $(\cT_t)_{t \ge 0}$ with generator $\mathcal{L}$ (i.e. $\cT_t=e^{t\mathcal{L}}$) is \emph{GNS-symmetric} with respect to $\sigma$ if for all $t \ge 0$,
\[
\tr\!\left( X\, \cT_t^*(Y)\, \sigma \right)
=
\tr\!\left( \cT_t^*(X)\, Y\, \sigma \right),
\qquad \forall\, X,Y \in B(H).
\]

\end{definition}

\begin{definition}[Entropy Contraction Coefficient\cite{GJLL}]
\label{def: alpha}
We say that a quantum channel $\Phi$ satisfies an \emph{$\alpha$-entropy contraction} for some $\alpha\ge 0$
\[
D\!\left( \Phi(\rho) \,\big\|\, \Phi\!\circ\!\mathcal{P}(\rho) \right)
\le e^{-\alpha} \,
D\!\left( \rho \,\big\|\, \mathcal{P}(\rho) \right),
\]
where $D(\rho \| \sigma) = \tr\left[ \rho(\log\rho - \log\sigma) \right]$ denotes the quantum relative entropy.
We denote by $\alpha(\Phi)$ (or simply $\alpha$ when $\Phi$ is clear from context) the largest constant $\alpha$ for which the above inequality holds.

For a quantum dynamical semigroup $\{\mathcal{T}_t\}_{t \ge 0}$ generated by a Lindbladian $\mathcal{L}$, the quantity $\alpha(\mathcal{L})$ is defined as the largest $\alpha$ such that, for all states $\rho$,
\[
D\!\left( \mathcal{T}_t(\rho) \,\big\|\, \mathcal{T}_t\!\circ\!\mathcal{P}(\rho) \right)
\le e^{-\alpha t} \,
D\!\left( \rho \,\big\|\, \mathcal{P}(\rho) \right),
\quad \forall\, t \ge 0.
\]
\end{definition}

\begin{definition}
\label{def: Lambda}
Let $\Phi$ be a GNS-symmetric quantum channel (or semigroup), and let $\mathcal{P}_\Phi$ denote its peripheral projection.  
The \emph{Pimsner--Popa constant} of $\mathcal{P}_\Phi$ is defined as
\begin{equation}
    \label{eq: def Lambda}
    \Lambda(\mathcal{P}_\Phi)
    := \inf \bigl\{\, C>0 \;\big|\; X \le C\,\mathcal{P}_\Phi(X) \ \text{for all } X \ge 0 \bigr\},
\end{equation}
which coincides with the Pimsner--Popa index of the associated conditional expectation \cite{pimsner1986entropy,GJLL}.  
In finite dimensions, $\Lambda(\mathcal{P}_\Phi)$ is always finite \cite[Remark 3.7]{GJLL}.
\end{definition}
To address tensorization, we introduce the complete contraction coefficient and complete Pimsner--Popa index
\[  \alpha_c(\Phi)=\inf_{n }\alpha(\idd_{\mathcal{M}_n}\otimes\Phi),\text{ and }  \quad\Lambda_{c}(\cP)=\sup_{n} \Lambda(\cP\otimes \idd_{\cM_n} ), \] 
where $\cM_n$ can be viewed as a reference system.

\begin{lemma}[Proposition 4.2 of \cite{GJLL}]
\label{lem: gns peripheral projection}
Let $\Phi$ be a GNS-symmetric quantum channel with respect to a faithful invariant state.  
Then the peripheral projection $\mathcal{P}$ associated with $\Phi$ is given by the  limit
\(\cP
= \lim_{t \to \infty, t\in \mathbb{N}} \Phi^{2t}.
\)
Similarly, if $(\cT_t)_{t \ge 0}$ is a GNS-symmetric quantum Markov semigroup with generator $\mathcal{L}$, then the peripheral projection is given by
\(
\mathcal{P}
= \lim_{t \to \infty, t\in \mathbb{R}} \cT_t.
\)

\end{lemma}
We note that the adjoint peripheral projection $\mathcal{P}^{*}$ is a conditional expectation; in particular, it is unital, completely positive, and idempotent.

\begin{lemma}[Remark 4.13 and Corollary 5.2 of \cite{GJLL}]
Let $\Phi$ be a GNS-symmetric quantum channel (or quantum Markov semigroup $\cT_t$) with respect to a faithful invariant state $\sigma$, and let $\mathcal{P}$ denote its peripheral projection. Then
\begin{equation}
\label{eq: alpha lower bound}
    \frac{\lambda}{\ln\!\bigl(10\,\Lambda_{c}(\mathcal{P})\bigr)}
    \;\le\;
    \alpha_{c},
\end{equation}
and
\begin{equation}
\label{eq: Lambda upperbound}
    \Lambda=\|\sigma^{-1}\|_\infty,\text{ and   }\quad \Lambda_c \le\|\sigma^{-1}\|_\infty^2,
\end{equation}

where, for a channel $\Phi$ and a quantum Markov semigroup $\mathcal{T}_t$, the parameter $\lambda$ is defined as the $L_2$-contraction coefficient:
 \begin{align}
 \label{eq: lambda for channel def}
 &e^{-\lambda(\Phi)}= \| \Phi^*(\mathrm{id}-\cP^*) \|_{\infty}  & \text{Definition 5.1 \cite{GJLL}} \\
 &e^{-\lambda(\cT) t}\ge \| \cT_t^*- \cP^* \|_{\infty}
    &\text{ Proposition 4.12 \cite{GJLL}}.
 \end{align}

\end{lemma}

\section{Main results}
\label{sec: main}
\subsection{Fast convergence to the infinite-time capacity}
We now state our main results formally for the one-shot and asymptotic capacities.
\begin{theorem}
\label{thm: single shot capacity}
     Let $\Phi$ be a GNS-symmetric quantum channel whose peripheral subspace has structure $\mathcal X_\Phi = \bigoplus_{k=1}^K \cM_{d_k}\otimes \omega_k$. Then the  one-shot capacities obey the following upper bounds: for $\delta < \frac{1}{2}$
     \begin{align}
     \label{eq:class-bound}
       %\log (\sum_{k=1} d_k)&\le
       C_\delta(\Phi^{2t})  \leq\frac{\log (\sum_{k=1}^K d_k)+ e^{-2\alpha_c t}\log(\Lambda_c(\cP))}{1-\delta}+h(\delta)
    \end{align}
    and for $\delta<\tfrac{1}{4}$
\begin{align}
    \label{eq:alt}
%\log (\max_{k} d_k)& \le
Q_\delta(\Phi^{2t}) \le \frac{ \log (\max_{k} d_k)+ e^{-2\alpha_c t}\log \Lambda_c(\cP)+2h(\delta)}{1-4\delta}.  
\end{align}

The same bounds hold for a quantum dynamical semigroup $\{\mathcal{T}_t\}_{t \ge 0}$, with $\Phi^t$ replaced by $\mathcal{T}_t$ throughout. Here $\alpha_c$ and $\Lambda_c$ are defined in Definitions~\ref{def: alpha} and~\ref{def: Lambda}, respectively.
\end{theorem}

\begin{theorem}
\label{thm: Capacity upper bound}
Let $\Phi : B(H) \to B(H)$ be a GNS-symmetric channel, and let $\alpha_c$ denote its complete entropy contraction constant. Let $\Lambda_c$ be defined as in~\eqref{eq: def Lambda}. 
For any finite $t$,
\begin{align*}
&\log (\sum_{k=1} d_k)\le C(\Phi^{2t}) \le \log (\sum_{k=1}^K d_k)  +e^{-2\alpha_c t}\log \Lambda_{c}(\cP)\\
&\log (\max_{k} d_k)\le Q(\Phi^{2t}) \le \log (\max_{k} d_k) +e^{-2\alpha_c t}\log \Lambda_{c}(\cP).
\end{align*}
The same bounds hold for a quantum dynamical semigroup $\{\mathcal{T}_t\}_{t \ge 0}$ 
with $\Phi^t$ replaced by $\mathcal{T}_t$ throughout.
\end{theorem}

To streamline the presentation, we first state the key lemmas underlying our argument.

\begin{lemma}\label{lemma:trace}
Let $\sigma_1, \sigma_2 \in \mathcal{X}_{\Phi}$ be states in peripheral subspace. Then for any $X$ such that $\cP(X)=0$,
\[\tr(X(\log\sigma_1-\log\sigma_2))=0.\]
\end{lemma}

\begin{proof}It follows from the general form of invariant states
\( \sigma=\bigoplus_{k} \sigma_k\otimes \omega_k\)
 that
\[ \log \sigma= \bigoplus_{k} \log (\sigma_k\otimes 1)+ \log(1\otimes \omega_k)\;.\]
Since $\omega_k$ are fixed states in $ B(H_{k,2})$, we have
\( \log {\sigma_1}-\log {\sigma_2}= \bigoplus_{k} Y_k\otimes 1\)
where the operators $Y_k=\log \sigma_{1,k}-\log \sigma_{2,k}$ are self-adjoint. ‌Since $\mathcal{P}$ has a full-rank fixed point, Theorem~5(\textbf{ii}, \textbf{iii}) of \cite{Blume_Kohout_2010} implies that the fixed-point subspace of $\mathcal{P}$—which coincides with the peripheral subspace of $\Phi$—is the distortion algebra associated with the fixed-point subspace of $\mathcal{P}^*$. (Note that since $\cP$ has a full-rank fixed point, the subspace $\cP_0$ defined in Theorem 5 of \cite{Blume_Kohout_2010} is the whole space).
In particular, if the peripheral subspace of $\Phi$ has the form
\(
\bigoplus_k \mathcal{M}_{d_k} \otimes \omega_k,
\)
then the fixed-point subspace of $\mathcal{P}^*$ (equivalently, $\mathrm{Range}(\mathcal{P}^*)$) has the structure
$
\bigoplus_k \mathcal{M}_{d_k} \otimes 1.
$
Consequently, operators of the form $\bigoplus_k Y_k \otimes 1$ belong to $\mathrm{Range}(\mathcal{P}^*)$. Therefore,
\begin{align*}
    \operatorname{tr}\!\left[X(\log \sigma_1 - \log \sigma_2)\right]
=  & \operatorname{tr}\!\left[X\, \mathcal{P}^*(\log \sigma_1 - \log \sigma_2)\right]
 \\ =&\operatorname{tr}\!\left[\mathcal{P}(X)(\log \sigma_1 - \log \sigma_2)\right]
= 0 .
\end{align*}

\end{proof}

\begin{lemma}
\label{lem: holevo coherent over time}
Let $\Phi : B(H) \to B(H)$ be a GNS-symmetric channel, and let $\alpha_c$ and $\Lambda_c$ be defined in Definitions~\ref{def: alpha} and~\ref{def: Lambda}, respectively. Let $\mathcal{P}$ denote the peripheral projection of $\Phi$. Then, for all $t > 0$,
\begin{align}
    \chi(\Phi^{2t}) - \chi(\mathcal{P}) &\le e^{-2\alpha t} \log \Lambda(\mathcal{P}), \\
    I_c(\Phi^{2t}) - I_c(\mathcal{P}) &\le e^{-2\alpha_c t} \log \Lambda_c(\mathcal{P}).
\end{align}
The same bounds hold for a quantum dynamical semigroup $\{\mathcal T_t\}_{t\ge 0}$ 
generated by a Lindbladian $\mathcal L$ 
satisfying an $\alpha_c$-complete entropy contraction, 
with $\Phi^{t}$ replaced by $\mathcal T_t$ throughout.

\end{lemma}

\begin{proof}
We begin with a simple observation: for any quantum channel $\Phi$ and any integer $t\ge 1$, the entropy contraction constants satisfy
$
\alpha(\Phi^{t}) \ge t\, \alpha(\Phi)$,  
$\alpha_c(\Phi^{t}) \ge t\, \alpha_c(\Phi)$.
Thus, it suffices to establish the desired bounds for the case $t=1$ (that is, for $\Phi^2$); the general case then follows immediately by replacing $\alpha_c$ with $t\alpha_c$ at the end.

\noindent\textbf{Holevo information.}
Let $\rho_{XB} = \sum_x p(x) \ketbra{x}{x} \otimes \rho_x$ be a classical--quantum state. Set
\[
\sigma_B := \Phi^2(\rho_B), \; \tau_B := \mathcal P(\rho_B),\;
\sigma_x := \Phi^2(\rho_x), \; \tau_x := \mathcal P(\rho_x),
\]
where $\rho_B := \sum_x p(x)\rho_x$. 
 Therefore,
\begin{align*}
\Delta I:=& I(X:B)_{(\mathrm{id}_X \otimes \Phi^2)(\rho)} - I(X:B)_{(\mathrm{id}_X \otimes \mathcal P)(\rho)} \\
&= \big(H(\sigma_B) - H(\tau_B)\big) - \sum_x p(x)\big(H(\sigma_x) - H(\tau_x)\big).
\end{align*}
For any states $\sigma,\tau$ on the same space, we have
\[
H(\sigma) - H(\tau) = -D(\sigma\|\tau) - \operatorname{tr}\big[(\sigma - \tau)\log \tau \big].
\]
Applying this to each term, we get
\begin{align*}
H(\sigma_B) - H(\tau_B)
&= -D(\sigma_B\|\tau_B) - \operatorname{tr}\big[(\sigma_B - \tau_B)\log \tau_B \big], \\
H(\sigma_x) - H(\tau_x)
&= -D(\sigma_x\|\tau_x) - \operatorname{tr}\big[(\sigma_x - \tau_x)\log \tau_x \big].
\end{align*}
We thus find
\begin{align*}
 \Delta I=&
 -D(\sigma_B\|\tau_B) + \sum_x p(x) D(\sigma_x\|\tau_x) \\
&- \operatorname{tr}\big[(\sigma_B - \tau_B)\log \tau_B \big]
   + \sum_x p(x)\operatorname{tr}\big[(\sigma_x - \tau_x)\log \tau_x \big].
\end{align*}

By Lemma~\ref{lem: gns peripheral projection}, for a GNS-symmetric channel for each $x$, we have 
\[
\mathcal P \circ \Phi^2 = \Phi^2 \circ \mathcal P = \mathcal P.
\]
Hence, for each $x$,
$
\mathcal P(\sigma_x - \tau_x) = 0.
$
Using Lemma~\ref{lemma:trace}, since both $\tau_x = \mathcal P(\rho_x)$ and $\tau_B = \mathcal P(\rho_B)$ lie in the range of $\mathcal P$, we can replace $\log\tau_x$ by $\log\tau_B$ under the trace:
\begin{align*}
    &\sum_x p(x)\operatorname{tr}\big[(\sigma_x - \tau_x)\log\tau_x\big]
\\ =&\sum_x p(x)\operatorname{tr}\big[(\sigma_x - \tau_x)\log\tau_B\big]
= \operatorname{tr}\big[(\sigma_B - \tau_B)\log\tau_B\big],
\end{align*}
because $\sigma_B = \sum_x p(x)\sigma_x$ and $\tau_B = \sum_x p(x)\tau_x$. The trace terms cancel, and by entropy contraction, we obtain
\begin{align*}
   &I(X:B)_{(\mathrm{id}_X \otimes \Phi^2)(\rho)} - I(X:B)_{(\mathrm{id}_X \otimes \mathcal P)(\rho)}\\
\le &e^{-2\alpha} \sum_x p(x) D\big(\rho_x \,\big\|\,\mathcal P(\rho_x)\big). 
\end{align*}
Using the uniform bound
\[
\sum_x p(x) D\big(\rho_x \,\big\|\,\mathcal P(\rho_x)\big) \le \log \Lambda(\mathcal P),
\]
we arrive at
\[
I(X:B)_{(\mathrm{id}_X \otimes \Phi^2)(\rho)} - I(X:B)_{(\mathrm{id}_X \otimes \mathcal P)(\rho)}
\le e^{-2\alpha} \log \Lambda(\mathcal P).
\]
Taking the supremum over all classical--quantum states $\rho_{XB}$ gives
\begin{align*}
\chi(\Phi^2)
%&= \sup_{\rho_{XB}} I(X:B)_{(\mathrm{id}_X \otimes \Phi^2)(\rho)} \\
&\le %\sup_{\rho_{XB}} I(X:B)_{(\mathrm{id}_X \otimes \mathcal P)(\rho)}
  % + e^{-2\alpha_c} \log \Lambda(\mathcal P) \\
 \chi(\mathcal P) + e^{-2\alpha_c} \log \Lambda(\mathcal P).
\end{align*}

\medskip
\noindent\textbf{Coherent information.}
Now let $\rho_{A'A}$ be any bipartite state and define
\[
\sigma_{A'B} := (\mathrm{id}_{A'} \otimes \Phi^2)(\rho_{A'A}), \quad
\tau_{A'B} := (\mathrm{id}_{A'} \otimes \mathcal P)(\rho_{A'A}),
\]
with marginals $\sigma_B := \operatorname{tr}_{A'}\sigma_{A'B}$ and $\tau_B := \operatorname{tr}_{A'}\tau_{A'B}$.

The coherent information is
$
I(A'\rangle B)_\omega = H(B)_\omega - H(A'B)_\omega.
$
Consider the difference
\begin{align*}
\Delta I
&:= I(A'\rangle B)_{\sigma} - I(A'\rangle B)_{\tau} \\
&= \big(H(B)_\sigma - H(A'B)_\sigma\big) - \big(H(B)_\tau - H(A'B)_\tau\big) \\
&= \big(H(B)_\sigma - H(B)_\tau\big) - \big(H(A'B)_\sigma - H(A'B)_\tau\big).
\end{align*}
Using again
\(
H(\sigma) - H(\tau) = -D(\sigma\|\tau) - \operatorname{tr}\big[(\sigma - \tau)\log \tau \big],
\)
similar to Holevo information 
we obtain
\begin{align*}
H(B)_\sigma - H(B)_\tau
&= -D(\sigma_B\|\tau_B) - \operatorname{tr}\big[(\sigma_B - \tau_B)\log \tau_B\big], \\
H(A'B)_\sigma - H(A'B)_\tau
&= -D(\sigma_{A'B}\|\tau_{A'B}) - \operatorname{tr}\big[(\sigma_{A'B} - \tau_{A'B})\log \tau_{A'B}\big].
\end{align*}
Substituting into $\Delta I$ gives
\begin{align*}
\Delta I
&= \Big[-D(\sigma_B\|\tau_B) - \operatorname{tr}\big[(\sigma_B - \tau_B)\log \tau_B\big]\Big] \\
&\quad - \Big[-D(\sigma_{A'B}\|\tau_{A'B}) - \operatorname{tr}\big[(\sigma_{A'B} - \tau_{A'B})\log \tau_{A'B}\big]\Big] \\
&= -D(\sigma_B\|\tau_B) + D(\sigma_{A'B}\|\tau_{A'B}) \\
&\quad - \operatorname{tr}\big[(\sigma_B - \tau_B)\log \tau_B\big]
   + \operatorname{tr}\big[(\sigma_{A'B} - \tau_{A'B})\log \tau_{A'B}\big].
\end{align*}

As before, from $\mathcal P\Phi^2 = \mathcal P$ we get
\[
(\mathrm{id}_{A'} \otimes \mathcal P)(\sigma_{A'B} - \tau_{A'B}) = 0,
\]
so $\sigma_{A'B} - \tau_{A'B}$ lies in the kernel of $\mathrm{id}_{A'}\otimes \mathcal P$. Since both
\[
\tau_{A'B} = (\mathrm{id}_{A'} \otimes \mathcal P)(\rho_{A'A})
\quad\text{and}\quad
1_{A'} \otimes \tau_B = 1_{A'} \otimes \mathcal P(\rho_B)
\]
belong to the range of $\mathrm{id}_{A'}\otimes\mathcal P$, Lemma~\ref{lemma:trace} implies
\[
\operatorname{tr}\big[(\sigma_{A'B} - \tau_{A'B})\log\tau_{A'B}\big]
= \operatorname{tr}\big[(\sigma_{A'B} - \tau_{A'B})(1_{A'}\otimes\log\tau_B)\big]
= \operatorname{tr}\big[(\sigma_B - \tau_B)\log\tau_B\big].
\]
The trace terms cancel, and we are left with
\[
\Delta I
= -D(\sigma_B\|\tau_B) + D(\sigma_{A'B}\|\tau_{A'B})
\le D(\sigma_{A'B}\|\tau_{A'B}),
\]
since $D(\sigma_B\|\tau_B)\ge 0$. Therefore,
\[
I(A'\rangle B)_{(\mathrm{id}_{A'} \otimes \Phi^2)(\rho)}
- I(A'\rangle B)_{(\mathrm{id}_{A'} \otimes \mathcal P)(\rho)}
\le D\big((\mathrm{id}_{A'} \otimes \Phi^2)(\rho)\,\big\|\,
(\mathrm{id}_{A'} \otimes \mathcal P)(\rho)\big).
\]

By complete entropy contraction,
\[
D\big((\mathrm{id}_{A'} \otimes \Phi^2)(\rho)\,\big\|\,
(\mathrm{id}_{A'} \otimes \mathcal P)(\rho)\big)
\le e^{-2\alpha_c} D\big(\rho\,\big\|\,(\mathrm{id}_{A'} \otimes \mathcal P)(\rho)\big),
\]
and by definition of $\Lambda_c(\mathcal P)$,
\[
D\big(\rho\,\big\|\,(\mathrm{id}_{A'} \otimes \mathcal P)(\rho)\big)
\le \log \Lambda_c(\mathcal P).
\]
Thus,
\[
I(A'\rangle B)_{(\mathrm{id}_{A'} \otimes \Phi^2)(\rho)}
- I(A'\rangle B)_{(\mathrm{id}_{A'} \otimes \mathcal P)(\rho)}
\le e^{-2\alpha_c} \log \Lambda_c(\mathcal P).
\]

Taking the supremum over all $\rho_{A'A}$ yields
\begin{align*}
I_c(\Phi^2)
&= \sup_{\rho_{A'A}} I(A'\rangle B)_{(\mathrm{id}_{A'} \otimes \Phi^2)(\rho)} \\
&\le \sup_{\rho_{A'A}} I(A'\rangle B)_{(\mathrm{id}_{A'} \otimes \mathcal P)(\rho)}
  + e^{-2\alpha_c} \log \Lambda_c(\mathcal P) \\
&\le I_c(\mathcal P) + e^{-2\alpha_c} \log \Lambda_c(\mathcal P).
\end{align*}

%\medskip
\noindent\textbf{Time scaling.}
We now show how the factor $e^{-2\alpha_c t}$ arises for general $t>0$.

For integer $t \ge 1$, by iterating the one-step entropy contraction and using $\mathcal P \circ \Phi = \mathcal P$, we have
\begin{align*}
D\big(\Phi^{t}(\rho)\,\big\|\,\mathcal P(\rho)\big)
&= D\big(\Phi^{t-1}(\Phi(\rho))\,\big\|\,\Phi^{t-1}(\mathcal P(\rho))\big) \\
&\le e^{-\alpha(\Phi)} D\big(\Phi^{t-1}(\rho)\,\big\|\,\mathcal P(\rho)\big) \\
&\le e^{-\alpha(\Phi)t} D\big(\rho\,\big\|\,\mathcal P(\rho)\big),
\end{align*}
where $\alpha(\Phi)$ is the (non-complete) entropy contraction constant. Thus $\alpha(\Phi^t) \ge t\,\alpha(\Phi)$, and similarly $\alpha_c(\Phi^t) \ge t\,\alpha_c(\Phi)$ for the complete constant. Hence, for all integer $t$,
\[
D\big( (\mathrm{id}_R \otimes \Phi^{2t})(\omega) \,\big\|\,
(\mathrm{id}_R \otimes \mathcal P)(\omega) \big)
\le e^{-2\alpha_c t} D\big( \omega \,\big\|\,
(\mathrm{id}_R \otimes \mathcal P)(\omega) \big),
\]
and the above proofs of the Holevo and coherent information bounds go through with $e^{-2\alpha_c}$ replaced by $e^{-2\alpha_c t}$.
This proves the stated bounds for both $\chi(\Phi^{2t})$ and $I_c(\Phi^{2t})$.

In the continuous-time case, the exact same proof can be applied.
\end{proof}

The following lemma follows easily from the results in~\cite{Wolfe2012}.
\begin{lemma}
\label{lem: chi and Ic of peripheral}
Let $\Phi$ be a quantum channel with peripheral projection $\cP_\Phi$.  
Suppose the peripheral subspace decomposes as  
\[
\mathrm{Ran}(\cP_\Phi)=\mathcal X_\Phi
=\bigoplus_k \cM_{d_k}\otimes \omega_k,
\]
where each $\omega_k$ is a full-rank state.  
Then
\[
\chi(\cP_\Phi)= \log\!\Big(\sum_k d_k\Big),
\qquad 
I_c(\cP_\Phi)= \log\!\big(\max_k d_k\big).
\]
\end{lemma}
\begin{proof}
Let  
\[
H = \bigoplus_k \big(H_{k,1}\otimes H_{k,2}\big),
\]
with $\dim H_{k,1}=d_k$, be the block decomposition of the peripheral subspace.  
The peripheral projection acts as  
\[
\cP(\rho)
= \bigoplus_k \big(\tr_{H_{k,2}}(\Pi_k\rho\Pi_k)\otimes \omega_k\big),
\]
so its image algebra is $\bigoplus_k \cM_{d_k}\otimes \omega_k$ (see e.g.,~\cite[Eq. 2.20]{amato2025asymptotic}).

We begin with the Holevo quantity.
For an ensemble $\{p_x,\rho_x\}$, each output state satisfies
\[
\sigma_x=\cP(\rho_x)
=\bigoplus_k \big(\sigma_{x,k}\otimes \omega_k\big),
\qquad
\sigma_{x,k}=\tr_{H_{k,2}}(\Pi_k\rho_x\Pi_k).
\]
Write $q_{x,k}=\tr(\sigma_{x,k})$ and $\tilde\sigma_{x,k}=\sigma_{x,k}/q_{x,k}$, and denote
$\bar\sigma=\sum_x p_x\sigma_x$.  
By orthogonality of blocks we have
\[
H(\sigma_x)= H(q_x)+\sum_k q_{x,k}\big(H(\tilde\sigma_{x,k}) + H(\omega_k)\big),
\]
\[
H(\bar\sigma)= H(\bar q)+\sum_k \bar q_k\big(H(\tilde{\bar\sigma}_k) + H(\omega_k)\big),
\]
where $\bar q_k=\sum_x p_x q_{x,k}$ and 
$\tilde{\bar\sigma}_k = \bar\sigma_k/\bar q_k$ with 
$\bar\sigma_k=\sum_x p_x\sigma_{x,k}$.

Therefore, the Holevo quantity of the ensemble $\{p_x,\rho_x\}$ is given by
\[
\big(H(\bar q)-\sum_x p_x H(q_x)\big)
 +\sum_k\!\left[\bar q_k H(\tilde{\bar\sigma}_k)
 -\sum_x p_x q_{x,k}H(\tilde\sigma_{x,k})\right].
\]
For each $k$ with $\bar q_k>0$,
\[
\tilde{\bar\sigma}_k
=\sum_x \frac{p_x q_{x,k}}{\bar q_k}\,\tilde\sigma_{x,k},
\]
so by concavity of entropy the term in brackets is nonnegative, and likewise
$H(\bar q)-\sum_x p_x H(q_x)\ge 0$.  Since $H(\tilde{\bar\sigma}_k)\le \log d_k$,
the Holevo quantity of the output ensemble $\{p_x,\sigma_x\}$ is bounded by
\[
H(\bar q)+\sum_k \bar q_k\log d_k.
\]

For any probability vector $(\bar q_k)$ and positive integers $(d_k)$,
\[
H(\bar q)+\sum_k \bar q_k\log d_k
= H(r) \le \log\!\Big(\sum_k d_k\Big),
\]
where $r_{k,i}=\bar q_k/d_k$ is a distribution on $\sum_k d_k$ outcomes.
Since this bound holds for every input ensemble $\{p_x,\rho_x\}$,
taking the supremum shows that the Holevo capacity satisfies
$\chi(\cP)\le \log\!\big(\sum_k d_k\big)$.
To see that equality is achieved, take an orthonormal basis 
$\{|k,i\rangle\}_{i=1}^{d_k}\subset H_{k,1}$ and any vectors $|\beta_k\rangle\in H_{k,2}$, and define an ensemble of size $N=\sum_k d_k$ by
\[
p_{k,i}=1/N,\qquad 
\rho_{k,i}=|k,i\rangle\!\langle k,i|\otimes|\beta_k\rangle\!\langle\beta_k|.
\]
These produce mutually orthogonal outputs  
$\sigma_{k,i}=|k,i\rangle\!\langle k,i|\otimes \omega_k$, with this choice, a direct calculation shows  
$\chi(\cP) \geq \log N$.  
Hence  
\[
\chi(\cP)=\log\!\Big(\sum_k d_k\Big).
\]

We now turn to the coherent information.  
Let $\cP^c$ be a complementary channel of $\cP$.  
Since $\cP$ is a direct sum of channels
\[
\cP_k:B(H_{k,1}\otimes H_{k,2})\to B(H_{k,1}\otimes H_{k,2}),
\qquad 
\cP_k(\rho)=\tr_{H_{k,2}}(\rho)\otimes\omega_k,
\]
we can choose the Stinespring dilations blockwise so that $\cP^c$ also decomposes as a direct sum,
\[
\cP^c(\rho)=\bigoplus_k \big(\rho_{k,2}\otimes \omega_k'\big),
\]
where $\rho_{k,2}=\tr_{H_{k,1}}(\Pi_k\rho\Pi_k)$ and $\omega_k'$ are fixed states (depending only on $k$) with
$H(\omega_k')=H(\omega_k)$ (they arise as the complementary marginals of a purification of~$\omega_k$).

Let $q_k=\tr(\Pi_k\rho\Pi_k)$, so $\sum_k q_k=1$.  Using the block-entropy formula, we obtain
\[
H(\cP(\rho))
=H(q)+\sum_k q_k\big(H(\rho_{k,1}/q_k)+H(\omega_k)\big),
\]
\[
H(\cP^c(\rho))
=H(q)+\sum_k q_k\big(H(\rho_{k,2}/q_k)+H(\omega_k)\big),
\]
where $\rho_{k,1}=\tr_{H_{k,2}}(\Pi_k\rho\Pi_k)$.
Subtracting gives an explicit expression for the coherent information:
\[
I_c(\cP,\rho)
= H(\cP(\rho)) - H(\cP^c(\rho))
=\sum_k q_k\big(H(\rho_{k,1}/q_k)-H(\rho_{k,2}/q_k)\big).
\]

Now $H(\rho_{k,1}/q_k)\le \log d_k$ because $\rho_{k,1}/q_k$ acts on $H_{k,1}$ of dimension $d_k$, and $H(\rho_{k,2}/q_k)\ge 0$, so  
\[
I_c(\cP,\rho)\le \sum_k q_k\log d_k\le \log(\max_k d_k).
\]
Thus $I_c(\cP)\le \log(\max_k d_k)$.

Finally, this upper bound is achievable.  
Pick an index $k^\star$ attaining $d_{k^\star}=\max_k d_k$, and take the state
\[
\rho=\Pi_{k^\star}\Big(\frac{I_{H_{k^\star,1}}}{d_{k^\star}}
 \otimes |\beta\rangle\!\langle\beta|\Big)\Pi_{k^\star},
\]
for any $|\beta\rangle\in H_{k^\star,2}$.  
Then $q_{k^\star}=1$ and all other $q_k$ vanish.  
In the expression above only the $k^\star$-block contributes, with
$H(\rho_{k^\star,1}/q_{k^\star})=\log d_{k^\star}$ and
$H(\rho_{k^\star,2}/q_{k^\star})=0$, so
\[
I_c(\cP,\rho)=\log d_{k^\star}.
\]
Therefore  
\[
I_c(\cP)=\log(\max_k d_k).
\]

This completes the proof.
\end{proof}

With these lemmas established, we can now turn to the proof of Theorem \ref{thm: single shot capacity}.

\begin{proof}[Proof of Theorem \ref{thm: single shot capacity}]
We first treat the classical capacity: using Fano's inequality and the Holevo bound for $\Phi^{2t}$, together with the mixing estimate for $\chi(\Phi^{2t})$ and the structure of the peripheral projection $\cP$, we obtain the bound on $C_\delta(\Phi^{2t})$. 

\textbf{Classical capacity.}
A $(d,\delta)$ classical code for $\Phi^{2t}$ is equivalent to the following: there is a uniform message $X\in[d]$ and a decoder’s guess $Y$, such that the average error
\[
P_e := \Pr\{Y\neq X\} \le \delta.
\]
 In fact, let the encoder images be
 \[
 \rho_i:=\cE(\ketbra{i}{i})\in B(H_A),\qquad
 \sigma_i:=\Phi^{2t}(\rho_i)\in B(H_B).
 \]
 Realize the decoder $\cD$ in the Heisenberg picture: the adjoint $\cD^\ast$ maps classical projectors to POVM elements
 \[
 M_i:=\cD^\ast(\ketbra{i}{i})\ge 0,\qquad
 \sum_{i=1}^d M_i=\mathbf 1_{B}.
 \]
 Measuring $\{M_i\}$ on $B$ produces the classical output $Y\in[d]$. The average success probability equals
 \[
 P_{\mathrm{succ}}
 = \frac{1}{d}\sum_{i=1}^{d}\tr(M_i\sigma_i)
 =F_c(\cD\circ \Phi^{2t}\circ \cE)\ \ge\ 1-\delta.
 \]

\emph{Fano's inequality.}
Let $E:=\mathbf 1_{\{X\neq Y\}}$ be the error indicator. Using the chain rule and $H(E|X,Y)=0$,
\[
H(X|Y)=H(E|Y)+H(X|E,Y).
\]
We bound each term:
\begin{align*}
H(E|Y)&\le H(E)=h(P_e)\ \le\ h(\delta),\\
H(X|E,Y)
&=\Pr\{E=0\} H(X|E=0,Y)\\&+\Pr\{E=1\} H(X|E=1,Y)\\
&\le (1-P_e)\cdot 0 + P_e \cdot \log(d-1)\ \le\ \delta \log(d-1).
\end{align*}
Therefore
\begin{equation}
H(X|Y)\ \le\ h(\delta)+\delta\log(d-1).
\label{eq:Fano}
\end{equation}
Since $X$ is uniform on $[d]$, $H(X)=\log d$, hence
\begin{equation}
I(X\!:\!Y)=H(X)-H(X|Y)\ \ge\ \log d - h(\delta)-\delta\log(d-1).
\label{eq:FanoMutual}
\end{equation}
Using $\log(d-1)\le \log d$ gives
\begin{equation}
I(X\!:\!Y)\ \ge\ (1-\delta)\log d - h(\delta).
\label{eq:FanoRelaxed}
\end{equation}

\textbf{Holevo bound via data processing.}
By data processing along $X\to B\to Y$,
\begin{equation}
I(X\!:\!Y)\ \le\ I(X\!:\!B)_{\rho_{XB}},
\label{eq:DP}
\end{equation}
where $\rho_{XB}=\frac{1}{d}\sum_i \ketbra{i}{i}\otimes \sigma_i$.
Since this is a cq-state obtained from a cq-input via $\Phi^{2t}$,
\begin{equation}
I(X\!:\!B)_{\rho_{XB}}\ \le\ \chi(\Phi^{2t}).
\label{eq:HolevoSup}
\end{equation}
Combining \eqref{eq:DP} and \eqref{eq:HolevoSup},
\begin{equation}
I(X\!:\!Y)\ \le\ \chi(\Phi^{2t}).
\label{eq:HolevoStep}
\end{equation}

From \eqref{eq:FanoRelaxed} and \eqref{eq:HolevoStep},
\[
(1-\delta)\log d - h(\delta)\ \le\ \chi(\Phi^{2t}),
\]
so
\[
\log d\ \le\ \frac{\chi(\Phi^{2t})+h(\delta)}{1-\delta}.
\]

By Lemma~\ref{lem: holevo coherent over time}, we know
\[
\chi(\Phi^{2t}) \le \chi(\cP) + e^{-2\alpha_c t}\log\Lambda(\cP).
\]
Using $\chi(\cP)=\log\big(\sum_k d_k\big)$ from Lemma~\ref{lem: chi and Ic of peripheral}, we obtain
\[
\chi(\Phi^{2t})
\le \log\Big(\sum_{k=1}^K d_k\Big) + e^{-2\alpha_c t}\log\Lambda(\cP).
\]
Since the above holds for every $(d,\delta)$ code for $\Phi^{2t}$, taking the supremum over $d$ yields
\[
C_\delta(\Phi^{2t})\ \le\ \frac{\log\big(\sum_{k=1}^K d_k\big) + h(\delta)
                              + e^{-2\alpha_c t}\log\Lambda(\cP)}{1-\delta},
\]
which is \eqref{eq:class-bound}.

\textbf{Quantum capacity.} The proof is similar and we use the quantum Fano inequality and the Araki--Lieb inequality to lower bound the coherent information of the overall code map, relate it to $I_c(\Phi^t)$ via data processing, upper bound $I_c(\Phi^t)$ by the mixing lemma and $I_c(\cP)$, and finally rearrange to obtain the  capacity bound \eqref{eq:alt}.

 \medskip
 Now suppose there exists a $(d,\delta)$ quantum code for $\Phi^{2t}$, with encoder $\cE$ and decoder $\cD$. Define the overall map
 \[
 N := \cD\circ \Phi^{2t}\circ \cE.
 \]
 Let $\tau_S$ be the maximally mixed state on a $d$-dimensional system $S$, and let $\phi_{RS}$ be a maximally entangled state on $R\otimes S$. Define
 \[
 \omega_{R\hat B} := (\idd_R\otimes N)(\phi_{RS}).
 \]
 By assumption on the code,
 \[
 F_e(\tau_S,N)\;=\;\bra{\phi}\omega_{R\hat B}\ket{\phi}\ \ge\ 1-\delta.
 \]
 By the quantum Fano inequality~\cite[Theorem 12.9.]{Nielsen_Chuang_2010},
 \begin{equation}
 \label{eq:quantum-fano}
 H(\omega_{R\hat B})
 %= H_e(\tau_S,N)
 \ \le\ h(\delta) + \delta\,\log(d^2-1).
 \end{equation}
 On the other hand, the Araki--Lieb inequality gives
 \(
 |H(R)-H(\hat B)| \leq H(R\hat B).
 \)
 Since $\omega_R$ is maximally mixed, $H(\omega_R)=\log d$, so
 \[
 H(\omega_{R\hat B}) \geq \log d - H(\omega_{\hat B}).
 \]
 Combining this with \eqref{eq:quantum-fano} we obtain
 \begin{equation}
 \label{eq:AL}
 H(\omega_{\hat B})\ \ge \log d
  -\big[h(\delta)+\delta\,\log(d^2-1)\big].
 \end{equation}

 Next, consider the coherent information of $N$ with maximally mixed input:
 \begin{equation}\label{eq:ic-identity}
 I_c(\tau_S,N)
 := H\big(N(\tau_S)\big) - H\big((\idd\otimes N)(\phi_{RS})\big)
  = H(\omega_{\hat B}) - H(\omega_{R\hat B}).
 \end{equation}
 Using \eqref{eq:quantum-fano} and \eqref{eq:AL},
 \[
 I_c(\tau_S,N)
 \ge \big[\log d - h(\delta)-\delta\log(d^2-1)\big]
     - \big[h(\delta)+\delta\log(d^2-1)\big],
 \]
 hence
 \begin{equation}
 \label{eq:IC-lower}
 I_c(\tau_S,N)
 \ \ge\ \log d - 2\big[h(\delta)+\delta\,\log(d^2-1)\big].
 \end{equation}

 We now relate this to the coherent information of $\Phi^{2t}$.  
 Let $\rho_A := \cE(\tau_S)$ be the average input to the channel.  
 By definition,
 \[
 I_c(\Phi^{2t}) := \sup_{\rho} I_c(\rho,\Phi^{2t}) \ \ge\ I_c(\rho_A,\Phi^{2t}).
 \]
 Moreover, $N=\cD\circ \Phi^{2t}\circ \cE$ differs from $\Phi^{2t}$ only by pre- and post-processing. Using the data-processing inequality for coherent information (equivalently, monotonicity of conditional entropy under local channels on $B$ and the encoding extension on the input), we have
 \[
 I_c(\rho_A,\Phi^{2t}) \ \ge\ I_c(\tau_S,N).
 \]
 Thus
 \begin{equation}
 \label{eq:IC-chain}
 I_c(\Phi^{2t})\ \ge\ I_c(\tau_S,N).
 \end{equation}

 Combining \eqref{eq:IC-lower} and \eqref{eq:IC-chain}, we obtain
 \[
 \log d \;\le\; I_c(\Phi^{2t}) + 2h(\delta) + 2\delta\,\log(d^2-1).
 \]

 Finally, we upper bound $I_c(\Phi^{2t})$ using the mixing lemma.  
 From Lemma~\ref{lem: holevo coherent over time}, we have
 \[
 I_c(\Phi^{2t}) - I_c(\cP)
 \;\le\;
 e^{-2\alpha_c t}\log \Lambda_c(\cP),
 \qquad \forall t>0.
 \]

 Using $I_c(\cP)=\log(\max_k d_k)$ from Lemma~\ref{lem: chi and Ic of peripheral}, we obtain
 \[
 I_c(\Phi^{2t}) \;\le\;
 \log\big(\max_k d_k\big)
 + e^{-2\alpha_c t}\log\Lambda_c(\cP).
 \]
 Substituting into the bound on $\log d$ gives :
 \[
 \log d
 \;\le\;
 \log \big(\max_{k} d_k\big)
 + e^{-2\alpha_c t}\log \Lambda_c(\cP)
 + 2h(\delta)
 + 2 \delta \log(d^2-1).
 \]

 If  $\delta<\tfrac{1}{4}$ and $Q_\delta(\Phi^{2t})= 0$ then \eqref{eq:alt} is trivial.
 If  $Q_\delta(\Phi^{2t})> 0$ (equivalently $d>1$), then
 \[
 \log(d^2-1) \le 2\log d,
 \]
 so 

 \[
 \log d \;\le\;
 \frac{\log(\max_k d_k)
       + e^{-2\alpha_c t}\log\Lambda_c(\cP)
       + 2h(\delta)}{1-4\delta}.
 \]
 Taking the supremum over such $d$ yields
 \[
 Q_\delta(\Phi^{2t})
 \;\le\;
 \frac{\log(\max_k d_k)
       + e^{-2\alpha_c t}\log\Lambda_c(\cP)
       + 2h(\delta)}{1-4\delta},
 \]
 which is \eqref{eq:alt}.

The semigroup case follows by replacing $\Phi^t$ with $\mathcal T_t$ and using the same entropy-contraction assumptions for the generator $\mathcal L$.
\end{proof}

We now present the proof of Theorem \ref{thm: Capacity upper bound}.
\begin{proof}[Proof of Theorem \ref{thm: Capacity upper bound}]
The asymptotic vanishing-error capacities of a channel $\Phi$ are given by
\begin{align*}
C(\Phi) &= \lim_{n\to\infty} \frac{1}{n}\,\chi(\Phi^{\otimes n}),\\
Q(\Phi) &= \lim_{n\to\infty} \frac{1}{n}\,I_c(\Phi^{\otimes n}).
\end{align*}

For a GNS-symmetric channel $\Phi$, it follows from~\cite{GR} that 
$\alpha_c(\Phi^{\otimes n}) = \alpha_c(\Phi)$.
Thus it remains to analyze how $\Lambda_c$ behaves under tensorization.

Let $\mathcal{P}: B(H_A)\to B(H_A)$ denote the peripheral projection of the channel $\Phi$.  
For any operator $Z\in B(H_A)\otimes B(H_A)$, we have
\begin{align*}
    Z 
    &\le \Lambda_c(\mathcal{P})\,(\mathcal{P}\otimes \mathrm{id}_A)(Z) \\
    &\le \Lambda_c(\mathcal{P})^2\,(\mathrm{id}_A\otimes \mathcal{P})\!\left((\mathcal{P}\otimes \mathrm{id}_A)(Z)\right) \\
    &= \Lambda_c(\mathcal{P})^2\,(\mathcal{P}\otimes \mathcal{P})(Z).
\end{align*}
Therefore,
\(
\Lambda_c(\mathcal{P}^{\otimes 2}) \le \Lambda_c(\mathcal{P})^2,
\)
and this inequality extends inductively to all $n>0$:
\(
\Lambda_c(\mathcal{P}^{\otimes n}) \le \Lambda_c(\mathcal{P})^n.
\)

Combining these facts with Lemma~\ref{lem: holevo coherent over time} completes the proof.
\end{proof}

\subsection{Zero-error  capacity}

The one-shot zero-error capacities $\mathcal{C}_0 \in \{C_0, Q_0\}$ quantify the amount of information a channel can transmit perfectly, without any chance of confusion between inputs. Their behavior is highly sensitive to the channel’s structure, making them fundamentally different from capacities allowing for error. In this subsection, we show how zero-error capacities behave under convex combinations and how, for GNS-symmetric channels, they stabilize to those of the limiting projection $\mathcal{P}$ once the channel has mixed sufficiently. The first step is the following monotonicity lemma, the proof of which is omitted.
\begin{lemma}
\label{lem: zero error convex}
If a channel $\Phi=\epsilon \cP+(1-\epsilon)\Psi$ is a convex combination of two channels $\cP$ and $\Psi$ with $\epsilon\in (0,1)$, then
\[\mathcal{C}_0(\Phi)\le \mathcal{C}_0(\cP).\]
Here $\mathcal{C}$ denotes either the classical capacity $C$ or the quantum capacity $Q$.
\end{lemma}
 \begin{proof}
 We argue for the quantum capacity $Q_0$ in the one-shot setting. By definition, $Q_0(\Phi)\ge \log d$ if there exists an encoder $\cE:\mathcal{M}_d\to B(H)$ and decoder $\cD:B(H)\to \cM_d$  
 such that
 \[ \idd\otimes \cD\Phi\cE(\ketbra{\phi}{\phi})=\ketbra{\phi}{\phi}\]
 where $\ketbra{\phi}{\phi}$ is the maximally entangled state on $\mathbb{C}^d\otimes \mathbb{C}^d$. Since
 $\Phi=\epsilon \cP+(1-\epsilon)\Psi$,
 we have
 \[ \epsilon\idd\otimes \cD\cP\cE(\ketbra{\phi}{\phi})+(1-\epsilon)\idd\otimes \cD\Psi\cE(\ketbra{\phi}{\phi})=\ketbra{\phi}{\phi}\]
 Since $\ketbra{\phi}{\phi}$ is a pure state, we know
 \[\idd\otimes \cD\cP\cE(\ketbra{\phi}{\phi})=\ketbra{\phi}{\phi}\]
 which implies
 \[ Q_0(\cP)\ge \log d\]
 Then $Q_0(\Phi)\le  Q_0(\cP)$, and the same inequality for the regularization implies
 \( Q_0(\Psi)\ge Q_0(\Phi). \)

 The same argument holds for $C_0$, as $C_0(\Phi)\ge \log d$ means that there exists an encoder $\cE:\cM_d\to B(H)$ and decoder $\cD:B(H)\to \cM_d$ such that
 \[  \cD\Phi\cE(\ketbra{i}{i})=\ketbra{i}{i}\]
 for an O.N.B $\{\ket{i}\}$ of $\mathbb{C}^d$.
 \end{proof}

Our main result in this section is the following.
\begin{theorem}
Let $\Phi$ be a GNS-symmetric channel with $\mathcal{P} = \lim_{t \to \infty} \Phi^{2t}$, and $\lambda$ defined as \eqref{eq: lambda for channel def}.  Then
\[
\mathcal{C}_0(\Phi^t) = \mathcal{C}_0(\mathcal{P}), \qquad 
\text{for all } t > \frac{\ln\bigl(10\,\Lambda_{c}(\mathcal{P})\bigr)}{\lambda},
\]
where the capacity $\mathcal{C}$ may be either $C$ or $Q$.
\end{theorem}

 \begin{proof}
 Proposition~4.12 of~\cite{GJLL} states that for any 
 \[
 t > \frac{\ln(10\Lambda_{c}(\mathcal{P}))}{\lambda},
 \]
 we have the completely positive order bounds
 \[
 0.9\,\mathcal{P} \;\le_{cp}\; \Phi^t \;\le_{cp}\; 1.1\,\mathcal{P}.
 \]
 Thus, for every such $t$,
 \[
 \Phi^t = (1 - \varepsilon)\mathcal{P} + \varepsilon\,\Psi,
 \]
 where $\Psi$ is a quantum channel and $\varepsilon \le 0.1$.

 By Lemma~\ref{lem: zero error convex}, this implies that for all  
 \[
 t > \frac{\ln(10\Lambda_{c}(\mathcal{P}))}{\lambda},
 \]
 the zero-error capacities satisfy
 \[
 C_0(\Phi^t) = C_0(\mathcal{P}), \qquad 
 Q_0(\Phi^t) = Q_0(\mathcal{P}).
 \]
 \end{proof}

\begin{remark}
To understand how the above bound behaves under tensorization, consider the channel $\Phi^{\otimes n}$. The spectral gap remains unchanged under tensor powers, i.e.\ $\lambda(\Phi^{\otimes n})=\lambda(\Phi)$. Moreover, as shown in the proof of Theorem~\ref{thm: Capacity upper bound}, the quantity $\Lambda_{c}$ appearing in Proposition~4.12 satisfies the tensorization property
\[
\Lambda_{c}\bigl(\mathcal{P}^{\otimes n}\bigr)
    = \Lambda_{c}(\mathcal{P})^{\,n}.
\]
Therefore, when applying the theorem to $\Phi^{\otimes n}$, the threshold
scales as
\[
\frac{\ln \bigl(10\,\Lambda_{c}(\mathcal{P})^{\,n}\bigr)}{\lambda}
= \frac{n\,\ln \bigl(\Lambda_{c}(\mathcal{P})\bigr)+\ln(10)}{\lambda}.
\]
Hence, the mixing-time bound grows linearly in the number of copies $n$, or equivalently, in the number of qubits. This improves the bound of \cite{singh2024zero}, where the zero-error capacity was shown to reach its limiting value only for $t \ge d^2$, with $d$ the Hilbert space dimension (which grows exponentially with the number of copies or qubits).

\end{remark}
\section{Example: active versus passive error correction}
\label{sec: examples}

In this section we illustrate how the passive upper bounds of 
Theorem~\ref{thm: Capacity upper bound} can be combined with lower bounds 
from concrete coding schemes to certify a strict advantage of active 
error correction over passive storage.

\subsection{Passive, active, and the $n$-active hierarchy}

By \textbf{passive} we mean that the noise channel $\Phi$ is applied 
repeatedly with no intervening operation: after $t$ time steps, the 
system has evolved under $\Phi^t$. By \textbf{active} we mean that a 
recovery channel is inserted between every two applications of $\Phi$. 
Within the active setting, it is natural to restrict how many copies of 
the system the recovery operation may touch at once.

\begin{definition}[$n$-active setup]
Let $\Phi : B(H) \to B(H)$ be a noise channel. In the 
\emph{$n$-active setup}, each recovery operator is a quantum channel
\[
\mathcal{R} : B(H^{\otimes n}) \to B(H^{\otimes n}),
\]
giving the decoder access to all $n$ copies of the system at each 
correction step. The case $n=1$ corresponds to recovery acting on a 
single copy. Thus, the capacity of $\Phi$ in the $n$-active setup at time $t$ is defined as
\[
Q_{\mathrm{act},n}^t(\Phi)
:= \frac{1}{n} \sup_{\mathcal{R}_1,\dots,\mathcal{R}_{t-1}} 
Q\!\left( \Phi^{\otimes n} \circ \mathcal{R}_{t-1} \circ \cdots \Phi^{\otimes n} \circ \mathcal{R}_1 \circ \Phi^{\otimes n} \right),
\]
where the maximization is over channels $\mathcal{R}_i : B(H^{\otimes n}) \to B(H^{\otimes n})$, and $Q$ is the asymptotic quantum capacity of the overall map. The classical $n$-active capacity $C_{\mathrm{act},n}^t(\Phi)$ is defined similarly, with $Q$ replaced by the classical capacity $C$.
\end{definition}

Since recovery maps with access to more copies can always simulate those 
with access to fewer, we obtain the inclusion chain
\[
1\text{-active} \,\subset\, 2\text{-active} \,\subset\, 3\text{-active} \,\subset\, \cdots
\]
The following lemma identifies the ceiling of this hierarchy: with 
unlimited active resources, the best achievable rate is just the  quantum capacity of $\Phi$, independent of the number of 
time steps. This makes $Q(\Phi)$-type quantities the natural benchmark 
against which passive performance should be compared.

\begin{lemma}
\label{lem: infty active}
The \emph{$\infty$-active capacity} of a channel $\Phi$ at time $t$, 
defined by
\[
Q_{\mathrm{act},\infty}^t(\Phi)
:= \lim_{n\to\infty} \frac{1}{n} 
\max_{\mathcal{R}_1,\dots,\mathcal{R}_{t-1}} 
Q\!\left( \Phi^{\otimes n} \circ \mathcal{R}_{t-1} \circ \cdots 
\circ \mathcal{R}_1 \circ \Phi^{\otimes n} \right),
\]
where the maximization is over channels 
$\mathcal{R}_i : B(H^{\otimes n}) \to B(H^{\otimes n})$, satisfies
\[
Q_{\mathrm{act},\infty}^t(\Phi) = Q(\Phi)
\qquad \text{for every finite } t \ge 1.
\]
\end{lemma}

\begin{proof}
\emph{Lower bound.} Fix $r < Q(\Phi)$ and $\varepsilon > 0$. By definition 
of the quantum capacity, there exist $n\ge 1$ and channels 
$\mathcal{E},\mathcal{D}$ such that
\[
\big\| \mathcal{D}\circ \Phi^{\otimes n}\circ \mathcal{E} 
- \mathrm{id}_2^{\otimes nr} \big\|_\diamond 
< \frac{\varepsilon}{t}.
\]
Setting $\mathcal{R}_i := \mathcal{E}\circ \mathcal{D}$ for $i = 1,\dots, t-1$ 
and composing with the outer encoder and decoder, the triangle 
inequality for the diamond norm gives
\[
\Big\| \mathcal{D}\circ \Phi^{\otimes n}\circ \mathcal{R}_{t-1}\circ \cdots 
\circ \mathcal{R}_1\circ \Phi^{\otimes n}\circ \mathcal{E} 
- \mathrm{id}_2^{\otimes nr} \Big\|_\diamond < \varepsilon,
\]
which shows $Q_{\mathrm{act},\infty}^t(\Phi) \ge r$. Taking 
$r \to Q(\Phi)$ yields $Q_{\mathrm{act},\infty}^t(\Phi) \ge Q(\Phi)$.

\emph{Upper bound.} For any choice of recovery channels 
$\mathcal{R}_1, \dots, \mathcal{R}_{t-1}$, the data-processing inequality 
applied to the single channel $\Phi^{\otimes n}$ at the first time step 
(with the remaining composition absorbed into the decoder) gives
$Q(\Phi^{\otimes n}\circ \mathcal{R}_{t-1}\circ \cdots \circ \Phi^{\otimes n}) 
\le Q(\Phi^{\otimes n}) = n\, Q(\Phi)$ by additivity.
Dividing by $n$ completes the proof.
\end{proof}

\subsection{A recipe for certifying active advantage}
\label{subsec: recipe}

The comparison between active and passive strategies proceeds in three 
steps.

\begin{enumerate}
\item \textbf{Passive upper bound.} Apply 
Theorem~\ref{thm: Capacity upper bound} to the GNS-symmetric noise 
channel $\Phi$. This gives the explicit bound
\[
Q(\Phi^{2t}) \;\le\; \log(\max_k d_k) + e^{-2\alpha_c t}\log \Lambda_c(\mathcal{P}),
\]
which controls any passive storage strategy at time $2t$. 

\item \textbf{Active lower bound.} Fix an $n$-qubit code with encoder 
$\mathcal{E}$ and decoder $\mathcal{D}$, and use 
$\mathcal{R} = \mathcal{E}\circ \mathcal{D}$ as the recovery operation. 
The effective logical channel per time step is 
$\mathcal{D}\circ \Phi^{\otimes n}\circ \mathcal{E}$; a lower bound on its 
quantum capacity (for instance, via hashing) yields an $n$-active lower 
bound for $\Phi^t$ after rescaling by $1/n$ to account for the physical 
qubit cost.

\item \textbf{Crossover time.} The smallest $t^\star$ at which the 
passive UB drops below the active LB certifies that active error 
correction strictly outperforms passive storage for all $t \ge t^\star$.
\end{enumerate}

We now execute this recipe on Pauli noise.

\subsection{Pauli noise and the 5-qubit code}

Consider the single-qubit Pauli channel
\[
\mathcal{T}_{\vec p}(\rho) 
= p_0\,\rho + p_x\,\sigma_x \rho \sigma_x 
+ p_y\,\sigma_y \rho \sigma_y + p_z\,\sigma_z \rho \sigma_z,
\]
with $\vec p = (p_0, p_x, p_y, p_z)$. The channel is GNS-symmetric with 
respect to the maximally mixed state, and its peripheral subspace is 
trivial whenever $\vec p$ has full support.

\paragraph{Passive upper bound.}
Inequalities~\eqref{eq: alpha lower bound} and~\eqref{eq: Lambda upperbound} 
yield $\Lambda_c(\mathcal{P}) \le 4$  and since the eigenvalues of $\mathcal{T}_{\vec p}$ are $1$ and $1 - 2(p_i + p_j)$ for $\{i,j\}\subset\{x,y,z\}$, we have
$\alpha_c \ge \lambda/\ln(10\,\Lambda_c)$, where
\[
\lambda 
= \min_{\{i,j\}\subset\{x,y,z\}} -\log\bigl| 1 - 2(p_i + p_j)\bigr|.
\]
Substituting into Theorem~\ref{thm: Capacity upper bound} gives an 
explicit passive upper bound on $Q(\mathcal{T}_{\vec p}^{2t})$ for every $t$.

A simple passive lower bound is also available: concatenating the channel with 
itself $t$ times yields another Pauli channel $\mathcal{T}_{\vec p_t}$ 
whose parameters $\vec p_t$ are straightforward to compute, and the hashing 
bound~\cite{bennett1996mixed}
\begin{equation}
\label{eq: hashing}
Q(\mathcal{T}_{\vec q}) \;\ge\; 1 - H(\vec q)
\end{equation}
applies directly.

\paragraph{5-active scheme from the 5-qubit code.}
Let $\mathcal{E}_5$ and $\mathcal{D}_5$ denote the encoder and decoder of 
the 5-qubit code $[[5,1,3]]$, and take the recovery operation to be 
$\mathcal{R} = \mathcal{E}_5\circ \mathcal{D}_5$. Inserting $\mathcal{R}$ 
between each application of $\mathcal{T}_{\vec p}^{\otimes 5}$ is 
equivalent, at the logical level, to a new Pauli channel
\[
\widetilde{\mathcal{T}}_{\vec q} 
\;:=\; \mathcal{D}_5\circ \mathcal{T}_{\vec p}^{\otimes 5}\circ \mathcal{E}_5,
\]
whose parameters $\vec q$ are a fixed polynomial function of $\vec p$ 
(computable from the weight enumerator of the code). Applying 
\eqref{eq: hashing} to $\widetilde{\mathcal{T}}_{\vec q}$ and rescaling by 
$1/5$ to account for the physical qubit cost yields a 5-active lower bound 
on $Q(\mathcal{T}_{\vec p}^{\,t})$.

\paragraph{Recursive concatenation: a 25-active scheme.}
Concatenating the 5-qubit code with itself produces a $[[25,1,9]]$-style 
code with encoder and decoder
\[
\mathcal{E}_{25} = \mathcal{E}_5^{\otimes 5}\circ \mathcal{E}_5,
\qquad
\mathcal{D}_{25} = \mathcal{D}_5 \circ \mathcal{D}_5^{\otimes 5}.
\]
Using $\mathcal{R} = \mathcal{E}_{25}\circ \mathcal{D}_{25}$ as the 
recovery operation gives a 25-active scheme, and the effective logical 
channel $\mathcal{D}_{25}\circ \mathcal{T}_{\vec p}^{\otimes 25}\circ \mathcal{E}_{25}$ 
is again Pauli. The hashing bound, rescaled by $1/25$, produces a 
stronger active LB at the cost of a larger block size. Concatenation 
therefore yields a family of $n$-active lower bounds indexed by the code 
depth, each trading physical qubits for a wider regime of provable 
active advantage.

\paragraph{Crossover.}
Figure~\ref{fig: pauli channel} plots all four bounds — the passive UB 
from Theorem~\ref{thm: Capacity upper bound}, the passive hashing LB, 
and the 5- and 25-active hashing LBs — for \\
$\vec p = (0.9986,\, 0.00047,\, 0.00047,\, 0.00047)$. Around 
$t \approx 8{,}000$, the 5-active LB crosses the passive UB, certifying 
that active error correction strictly outperforms any passive strategy 
beyond this point; the 25-active LB pushes the crossover substantially 
earlier. Tighter passive UBs (for example those specialized to Pauli 
channels, discussed below) would shift the crossover further to the 
left, strengthening the conclusion.

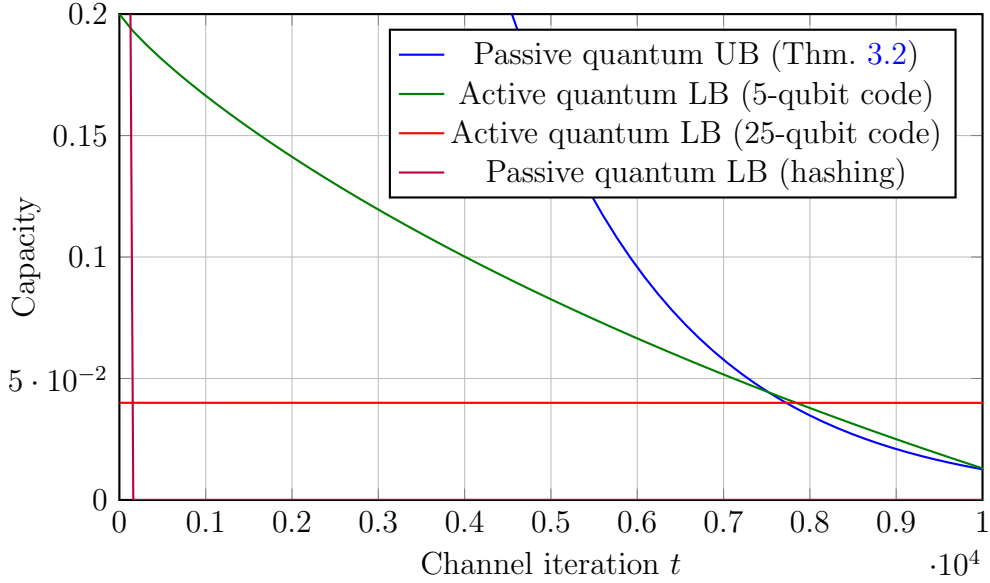
\begin{figure}
\centering
\begin{tikzpicture}
\begin{axis}[
    xlabel={Channel iteration $t$},
    ylabel={Capacity},
    xmin=0, xmax=10000,
    ymin=0, ymax=0.2,
    grid=major,
    legend pos=north east,
    width=13cm, height=8cm,
    thick,
    each nth point=40,
]
\addplot[color=blue]
table[col sep=comma, x index=0, y index=8]
{figs/plot_data/pauli_channel_capacities_0.9986_0.00047_0.00047_0.00047_T20000.csv};
\addlegendentry{Passive quantum UB (Thm.~\ref{thm: Capacity upper bound})}

\addplot[color=green!50!black]
table[col sep=comma, x index=0, y index=6]
{figs/plot_data/pauli_channel_capacities_0.9986_0.00047_0.00047_0.00047_T20000.csv};
\addlegendentry{Active quantum LB (5-qubit code)}

\addplot[color=red]
table[col sep=comma, x index=0, y index=7]
{figs/plot_data/pauli_channel_capacities_0.9986_0.00047_0.00047_0.00047_T20000.csv};
\addlegendentry{Active quantum LB (25-qubit code)}

\addplot[color=purple]
table[col sep=comma, x index=0, y index=5]
{figs/plot_data/pauli_channel_capacities_0.9986_0.00047_0.00047_0.00047_T20000.csv};
\addlegendentry{Passive quantum LB (hashing)}
\end{axis}
\end{tikzpicture}
\caption{Passive and active bounds on the quantum capacity of the 
iterated Pauli channel with 
$\vec p = (0.9986,\, 0.00047,\, 0.00047,\, 0.00047)$. The passive upper 
bound is obtained from Theorem~\ref{thm: Capacity upper bound}; the 
passive lower bound is the hashing bound applied directly to 
$\mathcal{T}_{\vec p}^{\,t}$; the 5-active and 25-active lower bounds 
are obtained by applying the hashing bound to the effective logical 
channels of the 5-qubit code and its recursive concatenation, rescaled 
by $1/5$ and $1/25$ respectively.}
\label{fig: pauli channel}
\end{figure}

\begin{remark}
Sharper upper bounds on the quantum capacity of Pauli 
channels~\cite{cerf2000pauli} would detect active advantage even at 
stronger noise levels, where our bound does not. The point of the 
present analysis is that the passive UB of 
Theorem~\ref{thm: Capacity upper bound} applies uniformly to all 
GNS-symmetric channels, at the cost of being loose on the Pauli 
subclass.
\end{remark}
\bibliography{refrence}
\bibliographystyle{plain}
\end{document}